\documentclass[journal]{IEEEtran}

\usepackage{epsfig}
\usepackage{amsmath}
\usepackage{amssymb}
\usepackage{amsthm}
\usepackage{multirow}
\usepackage{graphicx} 
\usepackage{blindtext}
\usepackage{xcolor}

\newtheorem{theorem}{Theorem}[section]

\newtheorem{proposition}[theorem]{Proposition}

\theoremstyle{definition}

\newcommand{\f}{\operatorname}

\ifCLASSINFOpdf
\else
\fi

\begin{document}

\title{Objective Bayesian Inference for Repairable System Subject to Competing Risks}

\author{Marco~Pollo, 
        Vera~Tomazella,
         Gustavo~Gilardoni, Pedro L. Ramos, Marcio J. Nicola and Francisco Louzada
\thanks{Marco Pollo with the Department
of Statistics, UFSCar-USP, S\~ao Carlos,
SP, Brazil, e-mail: mpa@usp.br. 
Vera Tomazella is with the Department of Statistics, UFSCar, S\~ao Carlos, SP, Brazil, e-mail: vera@ufscar.br. 
Gilardoni with the Department of Statistics, UNB, Brasilia, Brazil. 
Pedro L. Ramos with the Institute of Mathematical Science and Computing,University of S\~ao Paulo, S\~ao Carlos,
SP, Brazil, e-mail: pedrolramos@usp.br.
Marcio J. Nicola with the Institute of Mathematical Science and Computing,University of S\~ao Paulo, S\~ao Carlos,
SP, Brazil, e-mail: marcio.nicola@icmc.usp.br.
Francisco Louzada is with the Institute of Mathematical Science and Computing,University of S\~ao Paulo, S\~ao Carlos,
SP, Brazil, e-mail: louzada@icmc.usp.br.}
\thanks{Manuscript received March 29, 2018; revised March 29, 2018.}}

\markboth{ }%
{Shell \MakeLowercase{\textit{et al.}}: Bare Demo of IEEEtran.cls for IEEE Journals}

\maketitle

\begin{abstract}
Competing risks models for a repairable system 
subject to several failure modes are discussed.
Under minimal repair, it is assumed that each failure mode has a power law intensity.
An orthogonal reparametrization is used 
to obtain an objective Bayesian prior which is invariant 
under relabelling of the failure modes.
The resulting posterior is a product of 
gamma distributions and has appealing properties: 
one-to-one invariance, consistent marginalization and 
consistent sampling properties. Moreover, the resulting 
Bayes estimators have closed-form expressions and 
are naturally unbiased for all the parameters of the model.
The methodology is applied in the analysis of 
(i) a previously 
unpublished dataset about recurrent failure history 
of a sugarcane harvester and 
(ii) records of automotive warranty claims introduced in 
\cite{somboonsavatdee2015statistical}.
A simulation study was carried out 
to study the efficiency of the methods proposed.
\end{abstract}


\section*{ACRONYMS AND ABBREVIATIONS}
\vspace{-0.25cm}
\begin{table}[ht]
{\normalsize
\begin{tabular}{l l } 
Acronyms &  \\
PDF & Probability density function. \\
ML & Maximum likelihood. \\
MLEs & Maximum likelihood estimates. \\
MAP & Maximum a Posteriori. \\ 
MRE & Mean relative error.  \\
MSE & Mean square error.  \\
CP & Coverage probability. \\ 
CI & Credibility intervals. \\
NHPP & Non-homogeneous Poisson process. \\
PLP & Power law process. \\
ABAO & As bad as old. \\
CMLE & Conditionaly Unbiased MLE
\end{tabular}}
\end{table}

\section*{NOTATION}
\vspace{-0.25cm}
\begin{table}[ht]
{\normalsize
\begin{tabular}{l l } 
$\lambda(\cdot)$ & Intensity function. \\
$\lambda_j(\cdot)$ & Cause-specific intensity function. \\
$\Lambda(\cdot)$ & Cumulative intensity function. \\
$\Lambda_j(\cdot)$ & Cause-specific Cumulative intensity function. \\
$N_j(\cdot)$ & Cause-specific counting process. \\
$t$ & Continuous random variable represent failure time\\
\end{tabular}}
\end{table}

\begin{table}[ht]
{\normalsize
\begin{tabular}{l l } 
$t_i$ & i-th failure time. \\
$t_{ij}$ & i-th failure time of the cause j. \\
$\beta_j$ & Shape parameter of the cause j. \\
$\mu_j$ & Scale parameter of the cause j. \\
$\alpha_j$ & Mean number of failures of the cause j. \\
$\boldsymbol\theta$ & General parameters vector. \\
$\hat\beta_{j}^{MLE}$ & MLE of $\beta_j$. \\
$n$ & Number of failures of the system. \\
$n_j$ & Number of failures of the cause j. \\
$L(\cdot)$ & Likelihood function. \\
$\ell(\cdot)$ & Log-likelihood function. \\
$\gamma(\cdot)$ & Probability density function of gamma distribution. \\
$H(\cdot)$ & Fisher information matrix. \\
$\pi^{J}(\cdot)$ & Jeffreys prior distribution. \\
$\pi^{R}(\cdot)$ & Reference prior distribution. \\
$\pi(\cdot | \cdot)$ & Posterior distribution. \\
$\hat\theta^{MAP}$ & Maximum a posteriori estimator. \\
$\delta$ &  Cause of failure indicator. \\
$E(\cdot)$ & Expectation. \\
$\mathbb{I}(\cdot)$ & Indicator function. \\
$\hat\beta_{j}^{Bayes}$ & Bayes Estimator of $\beta_j$. \\
\end{tabular}}
\end{table}

\IEEEpeerreviewmaketitle

\vspace{0.5cm}

\section{Introduction}

\vspace{0.3cm}

\IEEEPARstart{T}{he} study of 
recurrent event data is important 
in many areas such as engineering and medicine. 
In the former, interest is usually centered in failure data 
from repairable systems. Efficient modeling and analysis 
of this data allows the operators of the equipments to 
design better maintenance policies, see 
\cite{escobar1998statistical, rigdon2000statistical,ascher1984repairable}.

In the repairable system literature, it is often assumed 
that failures occur following a NHPP with power law intensity. 
The resulting process is usually referred to as PLP.
The PLP process is convenient because it is 
easy to implement, flexible and the 
parameters have nice interpretation. Regarding classical 
inference for the PLP, see, for instance, 
Ascher and Feingold \cite{ascher1984repairable} or 
Ridgon and Basu \cite{rigdon2000statistical}.
Bayesian inference has been considered 
among others by 
Bar-Lev et al. \cite{bar1992bayesian}, Guida et al. \cite{guida1989bayes}, Pievatolo and Ruggeri \cite{pievatolo2004bayesian} and Ruggeri \cite{ruggeri2006reliability}. In this direction,
Oliveira et al. \cite{de2012bayesian} introduced 
an orthogonal parametrization 
(\cite{cox1987parameter, pawitan2001all})
of the PLP which 
simplifies both the analysis and the interpretation of the results.

In complex systems, such as supercomputers, 
airplanes or cars in the reliability area,
it is common the presence of multiple causes of failure. Traditionally, models with this characteristic are known as competing risks and are analyzed as if the causes had an independent structure, or analogously, a system where their $p$ components are connected in series, that is, when a component fails the whole system fails. In this paper we advocate the use of cause-specific intensity functions because they are observable quantities (unlike of the latent failure time approach) in competing risk analysis. See Pintilie \cite{pintilie2006competing}, Crowder et al. \cite{crowder2001classical} and Lawless \cite{lawless2011statistical}
for an overview about this approach.
 
Recently, Somboonsavatdee and Sen \cite{somboonsavatdee2015statistical} 
discussed classical inference for a repairable system that is subjected to multiple causes. 
Fu et al. \cite{fu2014objective} conducted a Bayesian reference analysis to model recurrent events with competing risks. They modeled failure types through a proportional intensity homogeneous Poisson process along with fixed covariates. In this framework, the objective of the authors was to estimate the cause-specific intensity functions. 
They used three reference priors according to different orders of grouping of model parameters. 
We note that, in the multiparamenter case, the overall reference prior 
may depend on which components of the vector 
of parameters are considered to be of primary interest. 
This problem can be quite challenging 
(see Berger et al. \cite{berger2015overall}, for a detailed discussion). 
Since we are interested simultaneously in all the parameters of the model, we propose an orthogonal reparametrization in the Somboonsavatdee and Sen \cite{somboonsavatdee2015statistical} model in order to obtain a unique overall reference prior.

The main purpose of this paper is to discuss Bayesian inference 
under an overall reference prior for the parameters of the PLP intensities 
under competing risks. 
We prove that the overall prior is also a matching prior, i.e., 
the resulting marginal posterior intervals have accurate frequentist coverage \cite{tibshirani1989}. The resulting posterior is proper and has interesting properties, such as one-to-one invariance, consistent marginalization, and consistent sampling properties. An extensive simulation study is presented which suggests that the resulting 
Bayes estimates outperform the estimates 
obtained from the classical approach. 

The rest of the paper is organized as follows: in Section II we present the motivating examples that will be used throughout of this paper. Section III presents some basic concepts regarding counting processes, repairable systems and competing risks model and objective Bayesian inference. Section IV presents the minimal repair for competing risks model and the maximum likelihood estimators of the model parameters. In Section V, we investigated s
Some properties of the data are characterized and methods of estimation are studied, from a single system assuming the independence of failure modes. In Section Z we study some results when there is a lack of prior information. The case of NHPP with a power law intensity function is studied in detail with an objective Bayesian inference approach. Finally, in section W, we present the conclusion and ideas for future research.

\section{Motivating Data}

Table \ref{tablelouzadar} 
shows failure times and causes for a sugarcane harvester during a crop. 
This machine harvest an average of 20 tons of sugarcane per hour and its malfunction can lead to major losses.
It can fail due to either malfunction of electrical components, 
of the engine or of the elevator, which are 
denoted as Cause 1, 2 and 3 respectively in Table \ref{tablelouzadar}.
There are 10, 24 and 14 failures
for each of these causes. 
During the harvest time that comprehends 254 days the machine operates on a 7x24 regime. 
Therefore, we will assume that each  
repair is minimal (i.e.\ it leaves the machine at exactly the 
same condition it was right before it failed) and that 
data collection time truncated at T = 254 days.  
The recurrence of failure causes is shown in 
Figure \ref{figurea}.

\begin{table}[!h]
\caption{Failure data for a sugarcane harvester}
\centering
{\scriptsize
\begin{tabular}{r|c|r|c|r|c|r|c}
\hline \hline 
Time & Cause & Time & Cause & Time & Cause & Time & Cause \\ 
\hline
 4.987 & 1 & 7.374 & 1 & 15.716 & 1 & 15.850 & 2 \\
20.776 & 2 & 27.476 & 3 & 29.913 & 1 & 42.747 & 1 \\
47.774 & 2 & 52.722 & 2 & 58.501 & 2 & 65.258 & 1 \\
71.590 & 2 & 79.108 & 2 & 79.688 & 1 & 79.794 & 3 \\
80.886 & 3 & 85.526 & 2 & 91.878 & 2 & 93.541 & 3 \\
94.209 & 3 & 96.234 & 2 & 101.606 & 3 & 103.567 & 2 \\
117.981 & 2 & 120.442 & 1 & 120.769 & 3 & 123.322 & 3 \\
124.158 & 2 & 126.097 & 2 & 137.071 & 2 & 142.037 & 3 \\
150.342 & 2 & 150.467 & 2 & 161.743 & 2 & 161.950 & 2 \\
162.399 & 3 & 185.381 & 1 & 193.435 & 3 & 205.935 & 1 \\
206.310 & 2 & 210.767 & 3 & 212.982 & 2 & 216.284 & 2 \\
219.019 & 2 & 222.831 & 2 & 233.826 & 3 & 234.641 & 3 \\
  \hline \hline 
\end{tabular}}\label{tablelouzadar}
\end{table}

\begin{figure}[!h] 
	\centering
	\includegraphics[width=8cm]{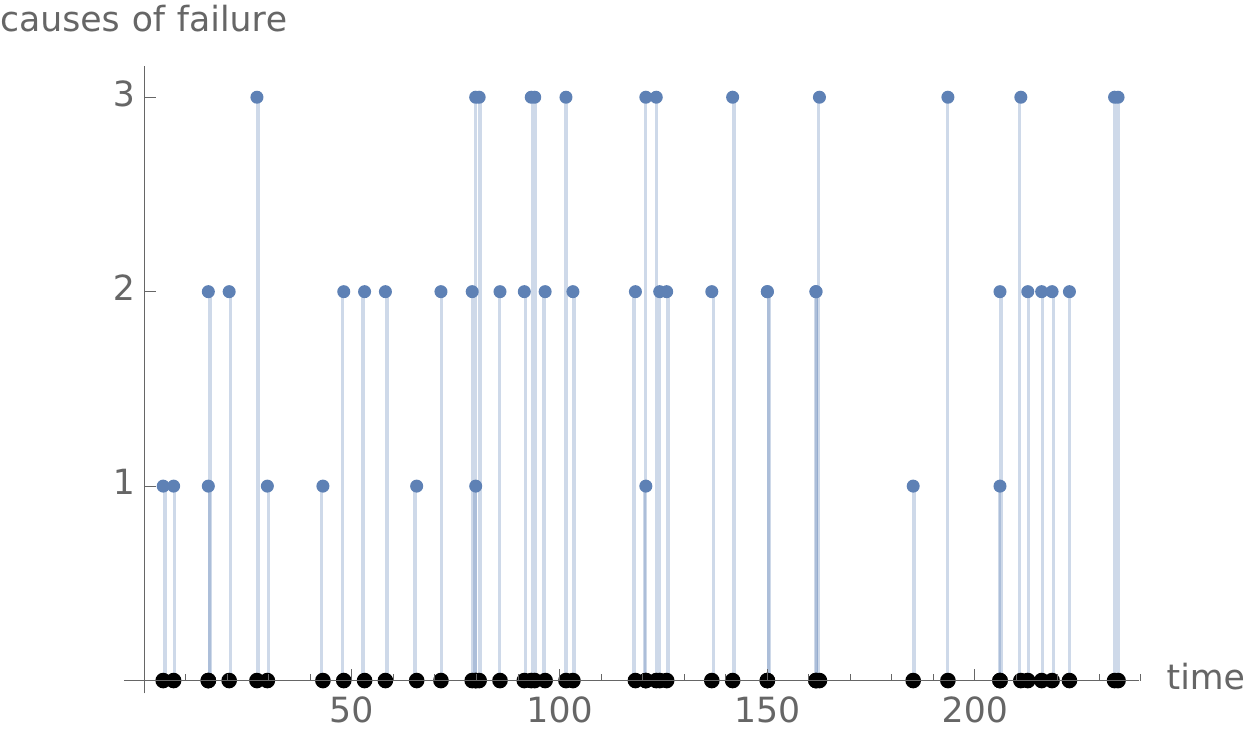}
	\caption{Recurrence of failures by cause and time. The black points on the x-axis indicate the system failures.  \label{figurea}}
\end{figure}

Our second data set was introduced by  
Somboonsavatdee and Sen \cite{somboonsavatdee2015statistical}.
It consists of warranty claims for a fleet of automobiles classified 
according to three possible causes.
The data describes the 372 recurrences of three causes of failures, hereafter, Cause 1, 2 and 3. 
The recurrence of failure causes can be seen in Figure \ref{figureb}. 
\begin{figure}[!h]
	\centering
	\includegraphics[width=9cm]{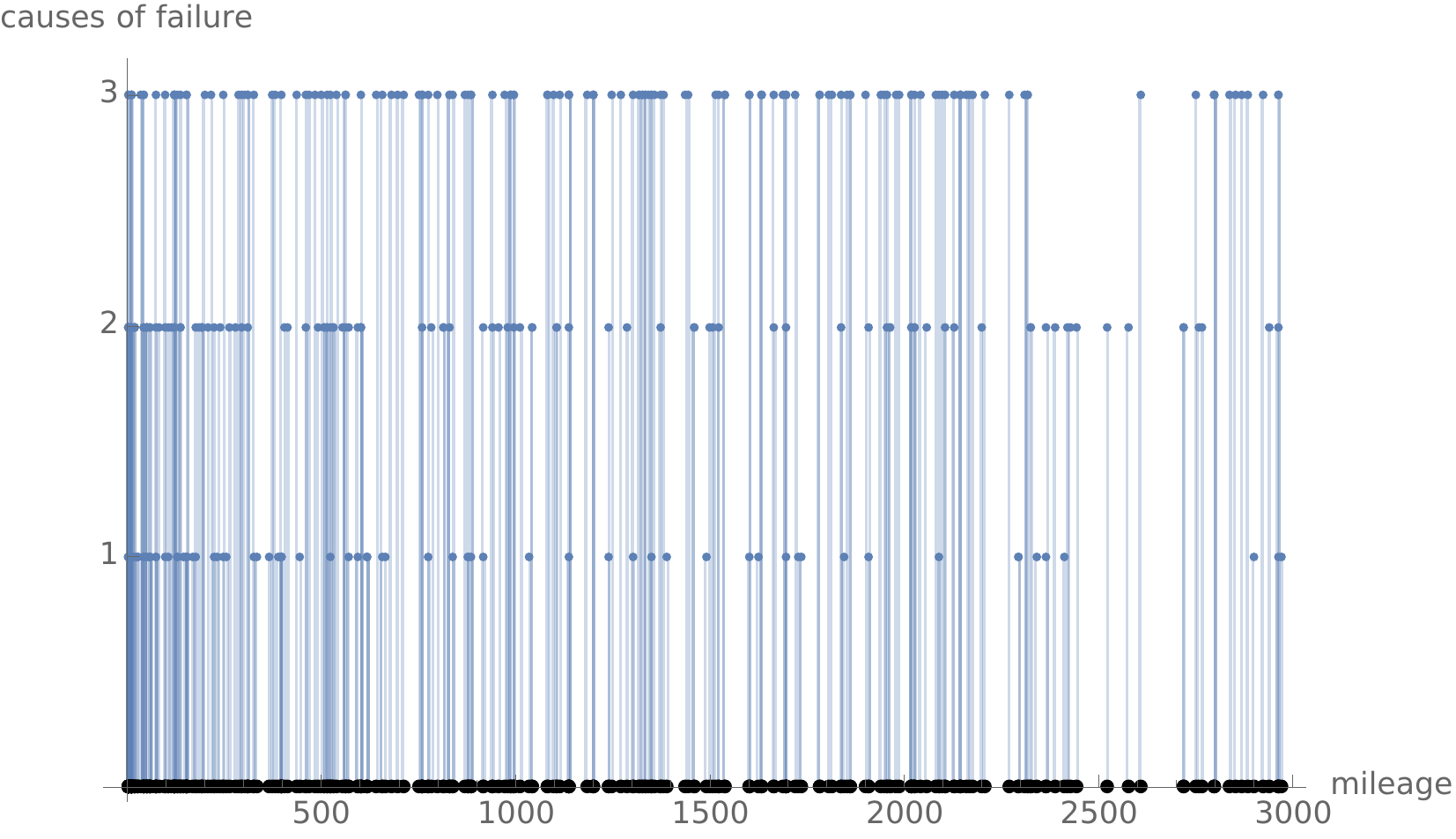}
	\caption{Recurrence of failures by cause and mileage. The black points on the x-axis indicates the system failures}\label{figureb}
\end{figure}	

Understanding the causes and the rate 
of accumulation of warranty claims very important in 
the automobile business. Warranty claim data can be used 
to estimate both the reliability of the vehicles in the field and 
the extent of future claim payments.
Hence, it will impact on 
costs, investment and security. 
For more details regarding a modern treatment of 
warranty claims please see \cite{xie2017two}, \cite{xie2016aggregate} and \cite{wang2017aggregate}.

\section{Background}

In this section, we discuss the analysis of the failures 
of a single unit repairable system under the action of 
multiple causes of recurrent and independent failures described by 
cause-specific intensity functions.

\subsection{The NHPP}
We begin by considering a repairable system 
with only one cause of failure. 
Let $N(t)$ be the number of failures before time $t$ and
$N(a,b] = N(b) - N(a)$ be the number of failures in the 
time interval $a < t \leq b$. 
An NHPP with intensity function $\lambda (t)$ ($t \geq 0$) 
is a counting process having independent increments 
and such that  
\begin{eqnarray}\label{equanhpp1}
\lambda (t)=\lim_{\Delta t \rightarrow 0} P(N(t, t+\Delta t]\geq 1) / \Delta t \,.
\end{eqnarray}
Alternatively, for each $t$, the random variate 
$N(t)$ follows a Poisson distribution with mean 
$\Lambda(t)=\int_{0}^{t} \lambda (s) ds$. 
A versatile parametric form for the intensity is
\begin{eqnarray}\label{equanhpp2}
\lambda (t)=(\beta / \mu) (t / \mu )^{\beta -1},
\end{eqnarray}
where $\mu, \beta>0$. In this case the NHPP is said 
to be a PLP and its mean function is 
\begin{eqnarray}\label{equanhpp3}
\Lambda(t)=E[N(t)]=\int_{0}^{t} \lambda (s) ds =(t / \mu )^{\beta}.
\end{eqnarray}
The scale parameter $\mu$ is the time for which we expect to observe a single event, while $\beta$ is the elasticity of the mean number of events with respect to time \cite{de2012bayesian}.

Since (\ref{equanhpp2}) is increasing (decreasing) in $t$
for $\beta > 1$ ($\beta < 1$), 
the PLP can accommodate both systems that deteriorate 
or improve with time. Of course, when $\beta =1$ the 
intensity (\ref{equanhpp2}) is constant and hence the PLP 
becomes a HPP.


\subsection{Competing Risks}
In reliability theory, the most common system configurations are the series systems, parallel systems, and series-parallel systems. 
In a series system, components are connected in such a way that the failure of a single component results in system failure 
(see Figure 3). 
A series system is also referred to as a competing risks system since the failure of a system can be classified as one of $p$ possible risks (componentS) that compete for the systems failure. 
\begin{figure}[!h]
	\centering
	\includegraphics[width=5cm]{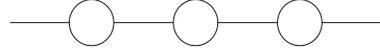}
	\caption{Diagram for a series system with three
components}
\end{figure}	

We note that we use the term {\em risks} before failure and 
{\em causes} afterward: 
the risks compete to be the cause. 
This paper is based on the classical assumption that the risks act independently. Thus, independent risks are equivalent to independent causes of failure.

A considerable amount of literature involving complex systems uses the assumption of stochastic independence in which is based on the physically independent functioning of components. 

Although competing risks is widely observed in medical studies, recent applications can bee seen in reliability. See, for instance,  
Crowder et. al \cite{crowder1994statistical} and references therein. 
For instance, an application can be seen when 
units under observation can experience any of several distinct failure causes, in which case for each unit we observe 
both the time to failure and the type of failure. 

\subsection{The Minimal Repair Model}

The main challenge when modelling repairable systems data 
is how to account for the effect of a repair action performed immediately 
after a failure. As usual, we will assume here that 
repair actions are instantaneous.
The most explored assumptions are either 
minimal repair and perfect repair. In the former, it is supposed 
that the repair action after a failure returns the system to 
exactly the same condition it was immediately before it failed. 
In the latter, the repair action leaves the system as if it were new. 
In the engineering literature this type of repair or of corrective 
maintenance are usually called ABAO and AGAN, for 
{\em as bad as old\/} and {as good as new\/} respectively
(\cite{barlow1960optimum, aven1983optimal, aven2000general,finkelstein2004minimal, mazzuchi1996bayesian}).
More sophisticated models which account for a repair 
action that leaves the system somewhere between the 
ABAO and AGAN conditions are possible, although they 
will not be considered here (see, for instance, 
\cite{doyen2004classes}).

Under minimal repair, the failure history of a repairable 
system is modelled as a NHPP. As mentioned above, 
the PLP (\ref{equanhpp2}) provides a flexible parametric 
form for the intensity of the process. 
Under the {\em time truncation} design, i.e.\ when failure data is 
collected up to time $T$, the likelihood becomes
\begin{eqnarray}
f(n, \boldsymbol{t}|\beta,\mu)=\frac{\beta^n}{\mu^{n\beta}} \left(\prod_{i=1}^{n} t_i \right)^{\beta -1}\hspace{-0.25cm}\exp\left[ - \left( \frac{T}{\mu}\right)^{\beta} \right],
\end{eqnarray}
where we assume that it has been observed $n \geq 1$ 
failures at times $t_1<t_2<\dotsc<t_n<T$ (see, for instance, 
Rigdon and Basu \cite{rigdon2000statistical}).  


Oliveira et al \cite{de2012bayesian}
suggest reparametrizing the model 
in terms of $\beta$ and $\alpha$, where
\begin{eqnarray}\label{reparam}
\alpha=E[N(T)]=(T/\mu)^{\beta},
\end{eqnarray}
so that the likelihood becomes
\begin{equation}\label{eqverplp}
\begin{aligned}
L(\beta, \alpha| n, \boldsymbol{t})&=c\left(\beta^n e^{n\beta/\hat{\beta}}\right)\left(\alpha^n e^{-\alpha}\right)  \\ &
\propto \gamma(\beta|n+1, n/\hat{\beta})\gamma(\alpha|n+1, 1)\,, 
\end{aligned}
\end{equation}
where $c=\prod_{i=1}^{n} t^{-1}_{i}$, $\hat{\beta}=n/\sum_{i=1}^{n} \log (T / t_i)$ is the maximum likelihood estimator of $\beta$ and $\gamma(x|a, b)=b^ax^{a-1}e^{-bx} / \Gamma(a) (x,a,b>0)$ is the PDF of the gamma distribution with shape and scale parameters $a$ and $b$, respectively. From (\ref{eqverplp}) it follows that $\beta$ and $\alpha$ are orthogonal (for the advantages of having orthogonal 
parameters see Cox and  Reid \cite{cox1987parameter}).

\subsection{Bayesian Inference}

The use of Bayesian methods has grown due to the advance of computational techniques. This approach is very attractive especially to construct credibility intervals for the parameters of the model. While in the classical approach the obtained intervals lie on the assumptions of asymptotic results, under the Bayesian approach such results can be obtained directly from the posterior density.

In this context, the prior distribution used to obtain the posterior quantities is of primary concern. Historical data or expert knowledge can be used to obtain a prior distribution. However, the elicitation process may be difficult and time consuming.
An alternative is to consider objective priors, in this case, we want to select a prior distribution in which the information provided by the data will not be obfuscated by subjective information.  Based on the above reasons, we focused on objective Bayesian analysis for the parameters of the cause-specific intensity functions. The formal rules to obtain such priors are discussed below.

\subsection{Jeffreys Prior}

Jeffreys \cite{jeffreys1961theory} proposed a rule for deriving a noninformative prior which is invariant to any one-to-one 
reparametrization. The Jeffreys' prior is one of the most popular objective priors and is proportional to the square root of the determinant of the expected Fisher information matrix $H(\boldsymbol\theta)$, i.e.,
\begin{equation*}
\pi^J(\boldsymbol\theta) \propto |H(\boldsymbol\theta)|^{1/2}. 
\end{equation*}

Although Jeffreys prior performs satisfactorily in one parameter cases, Jeffreys himself noticed that it may not adequate for the 
multi-parameter case. Indeed, such prior can lead to marginalization paradoxes and strong
inconsistencies (see Bernardo \cite[pg. 41]{bernardo2005reference}).

\subsection{Reference Prior}

Bernardo \cite{bernardo1979reference} introduced a class of objective priors known as reference priors. Such class of priors maximize the expected Kullback-Leibler divergence between the posterior distribution and the prior. The reference prior has minimal influence in a precise information-theoretic sense. that separated the parameters into the parameters of interest and nuisance parameters. To derive the reference prior function one need to set the parameters according to their order of inferential importance (see for instance, \cite{bernardo1979reference} and \cite{bernardo2005reference}). The main problem is that different ordering of the parameters return different priors and the selection of the more adequate prior may be quite challenging.

To overcome this problem Berger et al. \cite{berger2015overall} discussed different procedures to construct overall reference prior for all parameters. Additionally, under certain conditions, such prior is unique in the sense of being the same regardless the ordering of the parameters. To obtain such prior the expected Fisher information matrix must have a diagonal structure. The following result can be used to obtain the overall reference prior.

\begin{theorem}\label{theoveralpri} [Berger et al. \cite{berger2015overall}] Suppose that the Fisher information matrix of $\boldsymbol\theta$ is of the form
\begin{equation*}
H(\boldsymbol\theta)=\f{diag}(f_1(\theta_1)g_1(\boldsymbol\theta_{-1}),\ldots,f_k(\theta_m)g_k(\boldsymbol\theta_{-k})),
\end{equation*}
where $\boldsymbol\theta_{-i} = (\theta_1, \dotsb, \theta_{i-1}, \theta_{i+1}, \dotsb, \theta_k)$, $\f{diag}$ is a diagonal matrix, $f_i(\cdot)$ is a positive function of $\theta_i$ and $g_i(\cdot)$ is a positive function of $\boldsymbol\theta_i$, for $i=1, \dots, k$. Then, the reference prior, for any chosen interest parameter and any ordering of nuisance parameters, is given by
	\begin{equation}
	\pi^R (\boldsymbol\theta) \propto \sqrt{f_1(\theta_1) \ldots f_k(\theta_m)}\,.
	\end{equation}
\end{theorem}

The reference posterior distribution has desirable theoretical properties such as invariance under one-to-one transformations of the parameters, consistency under marginalization and consistent sampling properties.

Therefore, instead of constructing a reference prior for each order of grouping of parameters, we find a strategy very well explored in Oliveira et al. \cite{de2012bayesian} and cited in some examples of Rigdon and Basu \cite{rigdon2000statistical}, which does a reparametrization using the Poisson process mean function. 



\subsection{Matching Priors} 

Researchers attempted to evaluate inferential procures with good coverage errors for the parameters. While the frequentist methods usually relies on asymptotic confidence intervals, under the Bayesian approach formal rules have been proposed to derive such estimators. Tibshirani \cite{tibshirani1989} discussed sufficient conditions to derive a class of non-informative priors $\pi(\theta_1, \theta_2)$ where $\theta_1$ is the parameter of interest so that the credible interval for $\theta_1$ has \textcolor[rgb]{0,0,0}{a} coverage error $O$($n^{-1}$) in the frequentist sense, i.e.,
\begin{equation}\label{matchingp}
P\left[\theta_1\leq\theta_1^{1-a}(\pi;X)|(\theta_1,\theta_2)\right]=1-a-O(n^{-1}),
\end{equation}
where $\theta_1^{1-a}(\pi;X)|(\theta_1,\theta_2)$ denote the $(1-a)$th quantile of the posterior distribution of $\theta_1$. The class of priors  satisfying (\ref{matchingp}) are known as matching priors \cite{datta2012probability}.

To obtain such priors, Tibshirani \cite{tibshirani1989} proposed to reparametrize the model in terms of the orthogonal parameters $(\omega, \zeta)$ in the sense discussed by Cox
\& Reid \cite{cox1987parameter}. That is, $I_{\omega,\zeta}(\omega, \zeta)=0$ for all $(\omega, \zeta)$, where $\omega$ is the parameter of interest and $\zeta$ is the orthogonal nuisance parameter. In this case, the matching priors are all priors of the form
\begin{equation}\label{matchingpt}
\pi(\omega, \zeta)=g(\zeta)\sqrt{I_{\omega\omega}(\omega,\zeta)},
\end{equation}
where $g$($\zeta$)$>0$ is an arbitrary function and
$I_{\omega\omega}(\omega$, $\zeta$) is the $\omega$ entry of the Fisher
Information Matrix. The same idea is applied to derive priors when there are a vector of nuisance parameters. In the present study, we considered an orthogonal reparametrization in order to obtain priors with matching priors properties.

\subsection{Bayesian Point Estimators} 

There are different types of Bayesian estimators, 
the three most commonly used are the posterior mean, the posterior 
mode and the posterior median. Here we considered the posterior mode, that is usually refer as MAP estimate, due to structure that has a simple closed-form expression and can be rewritten as a
bias corrected MLE. A similar approach considering MAP estimates has been considered by Ramos et al. \cite{ramos2016efficient} to derive nearly unbiased estimator for the Nakagami distribution. One can define a maximum  a posteriori  estimator, $\hat{\boldsymbol\theta}^{MAP}$, which  is  given  by  maximizing the  posterior distribution, i.e.,
\begin{equation}
\begin{aligned}
\hat{\boldsymbol\theta}^{MAP} &= \operatorname*{arg\,max}_{\boldsymbol\theta} \pi(\boldsymbol\theta|\boldsymbol t) \\
&= \operatorname*{arg\,max}_{\boldsymbol\theta} \prod_{t_i \in \chi} f(t_i | \boldsymbol\theta) \pi(\boldsymbol\theta)  \\
&= \operatorname*{arg\,max}_\theta \left( \log( \pi(\boldsymbol\theta) ) + \sum_{t_i \in \mathcal{T}} \log( f(t_i | \boldsymbol\theta) )  \right).
\end{aligned}
\end{equation}

\section{Modeling Minimal Repair with Competing Risks}

The assumptions of the repairable system under examination is that the components can perform different operations, and thus be subject to different types of failures. Hence, in our model there are $p$ causes of failure. At each failure, the cause of failure (or failure mode) is denoted by $\delta_i=j $ for $j = 1, 2, \dots, p$.  If $n$ failures have been observed in $(0, T]$, then we observe the data $(t_1, \delta_1), \dots, (t_n, \delta_p)$, where $0 <t_1 <\dots < t_n < T$ are the system failure times and the $ \delta_i$s represent the associated failure cause with $i$-th failure time. 

One can introduce a counting process $N_j$ whose behavior is associated with the cause-specific intensity function
\begin{equation}
\lambda_j (t) = \lim_{\Delta t \rightarrow 0} \frac{P(\delta_i=j, N(t+\Delta t]\geq 1)}{\Delta t}.
\end{equation}
Hence, the global system failure process 
$N (t) =\sum_{j=1}^{p} N_j (t)$ is a superposition of NHPPs 
and its intensity is
\begin{equation}
\lambda (t) = \sum_{j=1}^{p} \lambda_j (t).
\end{equation}
The cause-specific and the system cumulative intensities are respectively
	\begin{equation}\label{intensCum}
	\Lambda_j (t)=\int_{0}^{t} \lambda_j(u) du \ \ \hbox{ and } \ \ \Lambda (t)= \sum_{j} \Lambda_j(t) ,
	\end{equation}
where we use the PLP intensity (\ref{equanhpp2}) 
for each cause, so that 
\begin{equation}
	\lambda_j (t)=\left( \frac{\beta_j}{\mu_j} \right) 
\left( \frac{t}{\mu_j}\right)^{\beta_j -1}
\end{equation}
for $j=1, \dots, p$ and
\begin{equation}\label{intensCum2}
\Lambda (t)= \sum_{j=1}^{p} \left( \frac{t}{\mu_j} \right)^{\beta_j}.
\end{equation}

\subsection{Maximum Likelihood Estimation}
Recalling that causes of failure act independently and they are mutually exclusive, so the general form of the likelihood function can be written as
\begin{equation}
\begin{aligned}
L(\boldsymbol\theta|\boldsymbol{t}) &= \prod_{i=1}^{n} \prod_{j=1}^{p} \{\lambda_j (t_i)\}^{\mathbb{I} ( \delta_i = j)} e^{ - \Lambda_j (t) } \\
&= \prod_{i=1}^{n} \{\lambda_1 (t_i)\}^{\mathbb{I} ( \delta_i = 1)} e^{ - \Lambda_1 (t) } \times \dots \times \\
& \quad  \{\lambda_p (t_i)\}^{\mathbb{I} ( \delta_i = p)} e^{ - \Lambda_p (t) }  \\
&= L_1 \times \dots \times L_p, 
\end{aligned}
\end{equation}	
where $\mathbb{I} ( \delta_i = j)$ represents the indicator function of the cause $j$ associated with $i-$th time of failure and 
$\boldsymbol\theta = (\mu_1 , \beta_1 , \ldots , \mu_p , \beta_p)$.	

To better understand how to compute the likelihood, we consider 
the following cases.

\subsection*{Case $p=2$, $\beta_1=\beta_2$}

Initially, we obtain the MLEs with only two independent failure causes, $j=1, 2$ and $\beta_1 = \beta_2 = \beta$ in which the system has been observed until a  fixed time $T$ . Resulting in
\begin{equation}\label{veros}
\begin{aligned}
L(\boldsymbol\theta|\boldsymbol{t}) &= \frac{ \beta^{n} }{ \mu_{1}^{n_1 \beta} \mu_{2}^{n_2 \beta} }   \left[ \prod_{}^{n_1} t_i  \prod_{}^{n_2} t_i \right]^{\beta - 1}\times \\
& \ \ \ \exp \left\lbrace - \left( \frac{T}{\mu_1} \right)^{\beta} - \left( \frac{T}{\mu_2} \right)^{\beta}  \right\rbrace,
\end{aligned}
\end{equation}	
where $\prod^{n_j} t_i = \prod_{i=1}^{n} t_i^{ \mathbb{I} (\delta_i = j)}$; $\sum_{i=1}^{n} \mathbb{I} ( \delta_i = j) = n_j $; $n = \sum_{j=1}^{p} n_j$ and $\boldsymbol\theta = (\beta, \mu_j)$.
The MLEs are
	\begin{equation}\label{mlebetaeq}
	\hat{\beta}=\frac{n}{\sum_{i=1}^{n} \log (T / t_i)} \quad \hbox{and} \quad \hat{\mu_j}=\frac{T}{n_{j}^{1/ \hat{\beta}}}.
	\end{equation}

\subsection*{Case $p=2$, $\beta_1\neq\beta_2$}

For the case of the different shape parameters, $\beta_1 \neq \beta_2$, the system intensity function is no longer a PLP. The likelihood function is given by
\begin{equation}
\begin{aligned}
L(\boldsymbol\theta|\boldsymbol{t}) &= \frac{\beta_1^{n_1} \beta_2^{n_2}}{\mu_{1}^{n_1 \beta_1} \mu_{2}^{n_2 \beta_2}}   \left[ \prod_{}^{n_1} t_i \right]^{\beta_1 - 1}  \left[ \prod_{}^{n_2} t_i \right]^{\beta_2 - 1} \\
& \quad \exp \left\lbrace - \left( \frac{T}{\mu_1} \right)^{\beta_1} - \left( \frac{T}{\mu_2} \right)^{\beta_2}  \right\rbrace,
\end{aligned}
\end{equation}	
where $\boldsymbol\theta = (\beta_1, \beta_2,\mu_1,\mu_2)$, and the MLEs are
\begin{equation}
	\hat{\beta_j}=\frac{n_j}{\sum_{i=1}^{n} \log (T / t_i) \mathbb{I} (\delta_i = j)} \quad \hbox{and} \quad \hat{\mu_j}=\frac{T}{n_{j}^{1/ \hat{\beta_j}}}.
	\end{equation}
Note that in the MLEs only exist if $n_j\geq1$ for $j=1,2$.

\subsection*{Case $p > 2$, at least two $\beta_j$s are different}

The likelihood function is given by
\begin{equation}
\begin{aligned}
L(\boldsymbol\theta|\boldsymbol{t}) &=  \prod_{j=1}^{p} \frac{\beta_j^{n_j}}{\mu_{j}^{n_j \beta_j}}   \left[ \prod_{}^{n_j} t_i \right]^{\beta_j - 1} \hspace{-0.4cm}\exp \left\{ -\sum_{j=1}^{p} \left( \frac{T}{\mu_j} \right)^{\beta_j}  \right\}
\end{aligned}
\end{equation}	
and the MLEs are
\begin{equation}
	\hat{\beta_j}=\frac{n_j}{\sum_{i=1}^{n} \log (T / t_i) \mathbb{I} (\delta_i = j)} \quad \hbox{and} \quad \hat{\mu_j}=\frac{T}{n_{j}^{1/ \hat{\beta_j}}}.
	\end{equation}
where again the MLEs exist only if $n_j\geq1$ for $j=1,2,\ldots,p$.

\section{Objective Bayesian Inference for the Model}

In this section, we present an objective Bayesian inference for the framework discussed so far by considering the reparametrization given in (\ref{reparam}) in order to obtain an orthogonal structure in the Fisher information matrix, and as a result, a unique objective prior.

\subsection*{Case $p =2$, $\beta_1 = \beta_2$}

Denote by $\beta$ the common value of $\beta_1$ and $\beta_2$.
The likelihood function considering the reparametrization is given by 
\begin{equation}\label{vercase2eq}
\begin{aligned}
L(\boldsymbol\theta) & =  \beta^{n} \left( T \alpha_1^{ - \frac{1}{\beta}} \right)^{-n_1 \beta} \left( T \alpha_2^{ - \frac{1}{\beta}} \right)^{-n_2 \beta}\times \\
& \quad \left[ \prod_{}^{n_1} t_i \prod_{}^{n_2} t_i \right]^{\beta - 1} e^{- \alpha_1 - \alpha_2 } \\
&\propto \gamma(\beta | n{+}1, n / \hat{\beta}) \prod_{j=1}^{2} \gamma( \alpha_j | n_j {+} 1, 1), 
\end{aligned}
\end{equation}
where now $\boldsymbol\theta=(\beta, \alpha_1, \alpha_2)$. The log-likelihood $\ell (\boldsymbol\theta)=log L(\boldsymbol\theta)$ is given by
\begin{eqnarray}
\ell (\boldsymbol\theta) \propto  n\log(\beta) - n\beta \log(T) + \beta \sum_{i=1}^{n} \log( t_i ) + \nonumber \\
+ n_1 \log( \alpha_1 ) + n_2 \log( \alpha_2 ) - \alpha_1 - \alpha_2. \nonumber
\end{eqnarray} 

The MLE for $\beta$ is the same as presented in (\ref{mlebetaeq}).
On the other hand, the MLEs for $\alpha_j$ are $\hat {\alpha}_{j}  = n_j$. 
To compute the Fisher Information Matrix, note that 
the partial derivatives are

$$\frac{\partial \ell}{\partial \beta} = n/ \beta - n\log(T) + \sum_{i=1}^{n} \log( t_i ),$$

$$\frac{\partial \ell}{\partial \alpha_j} = \frac{n_j}{\alpha_j} - 1,$$

$$\frac{\partial^2 \ell}{\partial \beta^2} = -n\beta^{-2},$$

$$\frac{\partial^2 \ell}{\partial \alpha_{j}^{2}} = \frac{-n_j}{\alpha_{j}^{2}} \ \ \ \mbox{and} $$

$$\frac{\partial^{2} \ell } {\partial \beta \partial \alpha_1 } = \frac{\partial^{2} \ell } {\partial \beta \partial \alpha_2 } = \frac{\partial^{2} \ell } {\partial \alpha_1 \partial \alpha_2 } = 0
$$
(note that, since we are considering time truncation, 
both $n = \sum_{j=1}^p n_j$ and the $n_j$'s are random). 
Hence, the expectation of the second derivatives above are given by
\begin{equation*}
-E\left(\frac{\partial^2 \ell}{\partial \beta^2}\right) = \frac{(\alpha_1 + \alpha_2)}{ \beta^{2}} \quad \mbox{and} \quad -E\left(\frac{\partial^2 \ell}{\partial \alpha_{i}^{2}}\right) = \frac{1}{\alpha_{j}}\cdot
\end{equation*}

The Fisher information matrix is diagonal and given by
$$
H(\boldsymbol\theta)= \left[
\begin{array}{ccc}
(\alpha_1 + \alpha_2) /\beta^{2}    & 0    &  0   \\
&     &     \\ 
0    &    1 / \alpha_{1}   & 0    \\
&     &        \\ 
0 & 0   &  1 / \alpha_{2}   \\  
\end{array}
\right].
$$

The first step in this approach begins by obtaining the Jeffreys prior distribution which is given by
\begin{equation}\label{jeffreyspcase2}
\pi^{J} (\boldsymbol\theta) \propto \frac{1}{\beta} \sqrt{\frac{\alpha_1 + \alpha_2}{\alpha_1 \alpha_2} }.
\end{equation}

\begin{proposition} The Jeffreys prior (\ref{jeffreyspcase2}) is a 
matching prior for $\beta$. 
\end{proposition}
\begin{proof} Let $\beta$ be the parameter of interest and
denote by $\boldsymbol\zeta=(\alpha_1,\alpha_2)$ the nuisance parameter. First, since the information matrix is diagonal, the 
Jeffreys' prior can be written in the form (\ref{matchingpt}) 
\end{proof}

The joint posterior distribution for $\beta$ and $\alpha_j$ produced by Jeffreys' prior is proportional to the product of the likelihood function (\ref{vercase2eq}) and the prior distribution (\ref{jeffreyspcase2}) resulting in
\begin{equation}\label{jeffreyspostcase2}
\begin{aligned}
\pi^J(\beta, \alpha_1, \alpha_2 | \boldsymbol t) \propto  & \sqrt{\alpha_1+\alpha_2}\left[\beta^{n-1} e^{-n\beta / \hat{\beta} }\right] \times  \\
&    \left[\alpha_1^{n_1 - 1/2} e^{-\alpha_1}\right]\left[\alpha_2^{n_2 - 1/2} e^{- \alpha_2}\right].
\end{aligned}
\end{equation}

The posterior (\ref{jeffreyspostcase2}) do not have closed-form, this implies that it may be improper, which is undesirable. Moreover, to obtain the necessary credibility intervals we would have to resort to Monte Carlo methods. To overcome these problems, we propose the alternative 
reference prior described below.

 based on .

Note that, if in Theorem \ref{theoveralpri} 
we take $ f_1(\beta) = \beta^{-2} $, $g_1(\alpha_1, \alpha_2)= (\alpha_1 + \alpha_2) $, $f_2(\alpha_1) = \alpha_{1}^{-1} $, $g_2(\beta, \alpha_2)= 1 $, $ f_3(\alpha_2) = \alpha_{2}^{-1} $ and $g_3(\beta, \alpha_1)= 1 $, the overall reference prior is 
\begin{equation}\label{refpcase2}
\pi^{R} ( \beta, \boldsymbol\alpha) \propto \frac{1}{\beta}\sqrt{\frac{1}{\alpha_1 \alpha_2}} .
\end{equation}

\begin{proposition}\label{overreprofmat} The overall reference prior (\ref{refpcase2}) is a matching prior for all the parameters. 
\end{proposition}
\begin{proof} If $\beta$ is the parameter of interest and $\zeta=(\alpha_1,\alpha_2)$, then the proof is analogous to 
that for the Jeffreys' prior above but considering $g(\boldsymbol\zeta)=\frac{1}{\sqrt{(\alpha_1+\alpha_2)\alpha_1\alpha_2}}$. 
If $\alpha_1$  is the parameter of interest and $\boldsymbol\zeta=(\beta,\alpha_2)$ are the nuisance parameters. Then, as $H_{\alpha_1,\alpha_1}(\alpha_1, \boldsymbol\zeta)=\dfrac{1}{\alpha_1}$ and $g(\boldsymbol\zeta)=\frac{1}{\beta\sqrt{\alpha_1}}$. Hence, the overall reference prior (\ref{refpcase2}) can be written in the form (\ref{matchingpt}). The case that $\alpha_2$ is the 
parameter of interest is similar. 
\end{proof}

The posterior distribution when using the 
overall reference prior (\ref{refpcase2}) is 
\begin{equation}
\begin{aligned}
\pi^R(\beta, \alpha_1, \alpha_2 | \boldsymbol t) \propto & \left[\beta^{n-1} e^{-n\beta / \hat{\beta} }\right]   \left[\alpha_1^{n_1 - 1/2} e^{-\alpha_1}\right]\times  \\
&  \left[\alpha_2^{n_2 - 1/2} e^{- \alpha_2}\right],
\end{aligned}
\end{equation}
that is,
\begin{equation}\label{postoverallbeq}
\pi^R(\beta, \alpha_1, \alpha_2 | \boldsymbol t) \propto  \gamma(\beta | n, n / \hat{\beta}) \prod_{j=1}^{2} \gamma( \alpha_j | n_j + 1/2, 1),
\end{equation}
which is the product of independent gamma densities. 
Clearly,  
if there is at least one failure for each cause, this posterior is proper.

The marginal posterior distributions are given by
\begin{equation}
\pi(\beta | \boldsymbol t) \propto \left[\beta^{n-1} e^{-n\beta / \hat{\beta} }
\right] \sim \gamma(\beta | n, n / \hat{\beta}), 
\end{equation}
and
\begin{equation}\label{margovbeq}
\pi( \alpha_i | \boldsymbol t) \propto \left[\alpha_i^{n_i - 1/2} e^{-\alpha_i}\right] \sim \gamma( \alpha_i | n_i + 1/2, 1).  
\end{equation}
Note that, as was proved in Proposition \ref{overreprofmat}, the marginal posterior intervals have accurate frequentist coverage for all parameters. 

From the posterior marginal
the Bayes estimator using the MAP for $\beta$ is given by
\begin{equation}
\hat{\beta}^{Bayes} = \left(\frac{n-1}{n}\right)\hat{\beta}^{MLE}.
\end{equation}

Throughout the rest of the paper we will denote 
$\hat{\boldsymbol\theta}^{Bayes}$  the Bayes estimators of $\boldsymbol\theta$. Rigdon and Basu \cite{rigdon2000statistical} argued that 
\begin{equation}\label{cmleeq}
E\left[\left(\frac{n-1}{n}\right)\hat{\beta}^{MLE} \bigg|\, n\right] = \beta,
\end{equation}
i.e., the proposed estimator is unbiased. Although they called CMLE, there is, however, no theoretical justification for its derivation. Here we provided a natural approach to obtain unbiased estimators for $\beta$ by considering the reference posterior.

In the case of Bayes estimators for $\alpha_j$, note that,
\begin{equation}\label{postjc2d}
\hat\alpha_j^{MEAN}= n_j +\frac{1}{2} \quad \mbox{and} \quad \hat\alpha_j^{MAP}= n_j-\frac{1}{2}.
\end{equation}

These estimators are biased specially when $n_j$ are small. On the other hand
\begin{equation}\label{postjc2e}
n_j-\frac{1}{2} <  n_j < n_j+\frac{1}{2}.
\end{equation}

Since $E(n_j)=\alpha_j$, this implies that as $n_j$ increase, $\hat\alpha^{MEAN}\approx\hat\alpha^{MAP}\approx n_j$. Therefore, as a Bayes estimator, we choose the unbiased estimator
\begin{equation}
\qquad \hat\alpha_j^{Bayes}=n_j, \qquad j=1,2.
\end{equation}

Ramos et al. \cite{ramos2018} presented a similar approach to obtain Bayes estimators for another distribution. It is noteworthy that the credibility interval must be evaluated considering the quantile function of the $\gamma( \alpha_i | n_i + 1/2, 1)$ to satisfy the matching prior properties.

\subsection*{Case $\beta_1 \neq \beta_2$}

In this case we substitute (\ref{reparam}) in (\ref{veros}) 
too obtain the likelihood function
\begin{equation}
\begin{aligned}
L(\boldsymbol\theta) &= c \left[ \beta_1^{n_1} e^{-n_1\beta_1/\hat{\beta_1}} \right] \left[ \beta_2^{n_2} e^{-n_2\beta_2/\hat{\beta_2}} \right] \times  \\ 
& \quad \ \left[ e^{-\alpha_1} \alpha_1^{n_1} \right]  \left[ e^{-\alpha_2} \alpha_2^{n_2} \right]\\
&\propto  \gamma(\beta_1 | n_1 + 1, n_1/{\hat{\beta_1}}) \gamma(\beta_2 | n_2 + 1, n_2/{\hat{\beta_2}}) \times  \\
& \ \ \ \ \gamma(\alpha_1 | n_1 + 1, 1) \gamma(\alpha_2 | n_2 + 1, 1),
\end{aligned}
\end{equation}
where $\boldsymbol\theta = (\beta_1, \beta_2, \alpha_1, \alpha_2)$ and $c=\left(\prod^{n_1} t_i \prod^{n_2} t_i\right)^{-1}.$ 
The MLEs have explicit solutions
\begin{equation}
\hat{\beta}^{MLE}_j =   \frac{n_j}{\sum_{i=1}^{n} \log (T / t_i) (\mathbb{I} ( \delta_i = j))} \quad  \hbox{and}  \quad \hat{\alpha}^{MLE}_{j} = n_j.
\end{equation}

Since $E[N_j(T)]=\alpha_j$, for $j=1,2$, the Fisher information matrix is
$$
H(\boldsymbol\theta)=\left[
\begin{array}{cccc}
\alpha_1  \beta_{1}^{-2} & 0 & 0 & 0 \\ 
&     &     &    \\
0 & \alpha_2  \beta_{2}^{-2} & 0 & 0 \\ 
&     &     &    \\  
0 & 0 & \alpha_{1}^{-1} & 0 \\ 

0 & 0 & 0 & \alpha_{2}^{-1} \\ 

\end{array}
\right].
$$ 

The Jeffreys prior is given by
\begin{equation}\label{priorjec2}
\pi^{J} (\boldsymbol\theta) \propto \frac{1}{\beta_1 \beta_2}.
\end{equation}

\begin{proposition}\label{propsjo2c} The Jeffreys prior (\ref{priorjec2}) is matching prior for $\beta_1$ and $\beta_2$. 
\end{proposition}
\begin{proof} Let $\beta_1$ be the parameter of interest and $\boldsymbol\lambda=(\beta_2,\alpha_1,\alpha_2)$ be  the nuisance parameters. Since the information matrix is diagonal
and $H_{\beta_1,\beta_1}(\beta_1, \boldsymbol\lambda)=\dfrac{\alpha_1}{\beta^2_1}$, taking $g(\boldsymbol\lambda)=\frac{1}{\beta_2\sqrt{\alpha_1}}$, (\ref{priorjec2}) can be written in the form (\ref{matchingpt}).
The case when $\beta_2$ is the parameter of interest is similar. 
\end{proof}

The joint posterior obtained using Jeffreys' prior (\ref{priorjec2}) is
\begin{equation}\label{postjec2j}
\pi^{J}(\boldsymbol\theta|\boldsymbol{t}) \propto \prod_{j=1}^{2} \gamma(\beta_j|n_j, n_j/\hat{\beta_j}(\boldsymbol{t})) \gamma(\alpha_j|n_j+1, 1)
\end{equation}
Therefore, the Bayes estimators using the MAP are
\begin{equation}
\hat\beta_j^{Bayes} = \left(\frac{n_j-1}{n_j}\right)\hat{\beta}^{MLE}_j \quad  \hbox{and}  \quad \hat {\alpha}^{Bayes}_{j}  = n_j.
\end{equation}

On the other hand, 
considering Theorem \ref{theoveralpri}, the overall prior distributions is
\begin{equation}\label{postRef}
\pi^{R} (\boldsymbol\theta) \propto \frac{1}{\beta_1 \beta_2} \frac{1}{\sqrt{\alpha_1 \alpha_2}} \cdot
\end{equation}

\begin{proposition}\label{propov2cmat} The overall reference prior (\ref{postRef}) is a matching prior for all parameters. 
\end{proposition}
\begin{proof} The proofs for $\beta_1$ and $\beta_2$ follow the same steps as in the proof of Proposition \ref{propsjo2c}. The cases of $\alpha_1$ and of $\alpha_2$ follow directly from Proposition \ref{overreprofmat}. 

\end{proof}

The reference posterior distribution is given by
\begin{equation}
\pi^{R}(\boldsymbol\theta|\boldsymbol{t}) \propto \prod_{j=1}^{2} \gamma(\beta_j|n_j, n_j/\hat{\beta_j}(\boldsymbol{t})) \gamma(\alpha_j|n_j+ \tfrac{1}{2}, 1) .
\end{equation}
Hence, the MAP estimator for $\beta_j$ is given by
\begin{equation}
\hat{\beta_j}^{Bayes} = \left(\frac{n_j-1}{n_j}\right)\hat{\beta_j}^{MLE}.
\end{equation}
To obtain the Bayes estimators for $\alpha_j$ we considered the same argument described in the last section, which gives in this case
\begin{equation}
\hat\alpha_j^{Bayes}=n_j.
\end{equation}

\subsection*{Case 3: p causes of failure and different $\beta$s}

The likelihood function in this case is
\begin{equation}
L(\boldsymbol\theta) \propto \prod_{j=1}^{p} \gamma(\beta_j | n_j + 1, n_j / \hat{\beta_j}) \gamma( \alpha_j | n_j + 1, 1) ,
\end{equation}
where $\boldsymbol\theta=(\beta_1,\ldots,\beta_p,\alpha_1,\ldots,\alpha_p)$. The MLEs have explicit solutions
\begin{equation}
\hat{\beta}^{MLE}_j =   \frac{n_j}{\sum_{i=1}^{n} \log (T / t_i) (\mathbb{I} ( \delta_i = j))} \quad  \hbox{and}  \quad \hat { \alpha}^{MLE}_{j} = n_j
\end{equation}
and the Fisher information matrix is
$$
H(\boldsymbol\theta) = \left[
\begin{array}{cccccc}
\alpha_1\beta_1^{-2} &  0  & \ldots & \ldots & \ldots & 0 \\
0  & \ddots & 0 & \ldots & \ldots & 0 \\
\vdots & 0 & \alpha_p\beta_p^{-2} & 0 & \ldots & 0 \\
\vdots  & \vdots & 0 & \alpha_1^{-1} & \ldots & 0 \\
\vdots & \vdots & \vdots & \vdots & \ddots & 0 \\
0 & \ldots & 0 & 0 & 0 & \alpha_p^{-1} \\
\end{array}
\right].
$$
The Jeffreys prior is
\begin{eqnarray}\label{priorjfp}
\pi^J(\boldsymbol\theta) \propto \prod_{j=1}^{p} \frac{1}{\beta_j} ,
\end{eqnarray}
which gives the posterior distribution 
\begin{eqnarray}\label{genpostjef}
\pi^J(\boldsymbol\theta | \boldsymbol t) \propto \prod_{j=1}^{p} \gamma(\beta_j | n_j, n_j / \hat{\beta_j}) \gamma( \alpha_j | n_j+1, 1).
\end{eqnarray}

To prove that  (\ref{priorjfp}) is a matching prior for $\beta_j, j=1,\ldots, p$ we can consider the same steps of the proof of 
Proposition \ref{propsjo2c}.

\begin{table*}[!t]
\centering
\caption{The MRE, MSE from the estimates considering different values of $n$  with $M=1.000,000$ simulated samples using the different estimation methods.}
\begin{tabular}{c|c|r|r|r|r|r|r|r|r|r|r}
\hline 
\hline
& & \multicolumn{2}{c|}{Scenario 1}  & \multicolumn{2}{c|}{Scenario 2}   & \multicolumn{2}{c|}{Scenario 3}   & \multicolumn{2}{c|}{Scenario 4} & \multicolumn{2}{c}{Scenario 5} \\
\hline

Parameter & Method & MRE & MSE   & MRE & MSE  & MRE & MSE  & MRE & MSE  & MRE & MSE  \\
\hline
           & MLE    & 1.2401  & 2.5087 & 1.0411 & 0.1484   & 1.2897 &  3.2159  & 1.2344 & 2.6494 & 1.1636 & 0.0255     \\ 
$\beta_1$  & Bayes  & 0.9991  & 0.9509 & 1.0000 & 0.1313   & 0.9993 &  1.1567  & 0.9995 & 1.0319 & 0.9993 & 0.0136   \\  \hline
           & MLE    & 1.5693  & 6.7586 & 1.5272 & 13.0190   & 1.0812 &  0.0740  & 1.0771 & 0.0527 & 1.0101 & 0.0426    \\ 
$\beta_2$  & Bayes  & 0.9980  & 1.8000 & 1.0008 & 3.4447   & 0.9998 &  0.0576   & 0.9998 & 0.0415 & 0.9999 & 0.0413     \\ \hline
           & MLE    & 1.0102  & 6.1686 & 0.9997 & 26.4448  & 1.0210 &  5.1477   & 1.0091 & 6.2929 & 1.0015 & 8.3588    \\ 
$\alpha_1$ & Bayes  & 1.0102  & 6.1686 & 0.9997 & 26.4448  & 1.0210 &  5.1477  & 1.0091 & 6.2929 & 1.0015 & 8.3588   \\ \hline
           & MLE    & 1.2306  & 2.2679 & 1.1700 & 2.5335  & 0.9997 & 14.4912   & 0.9998 & 15.1111 & 0.9999 & 99.8225     \\
$\alpha_2$ & Bayes  & 1.2306  & 2.2679 & 1.1700 & 2.5335  & 0.9997 & 14.4912   & 0.9998 & 15.1111 & 0.9999 & 99.8225   \\ \hline \hline 
\end{tabular}
\label{tableres1a}
\end{table*}

\begin{table*}[!t]
\centering
\caption{Coverage probabilities from the estimates considering different scenarios  with $M=1.000,000$ simulated samples and different estimation methods.}
\begin{tabular}{c|c|c|c|c|c|c}
\hline 
\hline
$\boldsymbol{\theta}$ & Method & Scenario 1 & Scenario 2   & Scenario 3 & Scenario 4 & Scenario 5  \\
\hline
           & MLE       & 0.9556 & 0.9523 & 0.9555 & 0.9555  & 0.9553     \\ 
           & CMLE      & 0.8740 & 0.9357 & 0.8599  & 0.8759  & 0.8958    \\ 
$\beta_1$  & Jeffreys  & 0.9503 & 0.9501 & 0.9503 & 0.9502  & 0.9503     \\ 
           & Reference & 0.9503 & 0.9501 & 0.9503 & 0.9502  & 0.9503     \\\hline
           & MLE       & 0.9538 & 0.9538 & 0.9538 & 0.9537  & 0.9501     \\ 
           & CMLE      & 0.7869 & 0.7989 & 0.9226 & 0.9237  & 0.9460     \\ 
$\beta_2$  & Jeffreys  & 0.9501 & 0.9498 & 0.9499 & 0.9499  & 0.9497     \\ 
           & Reference & 0.9501 & 0.9498 & 0.9499 & 0.9499  & 0.9497     \\\hline
           & MLE       & 0.8884 & 0.9325 & 0.9352 & 0.8943  & 0.9196     \\ 
           & CMLE      & 0.8884 & 0.9325 & 0.9352 & 0.8943  & 0.9196     \\ 
$\alpha_1$ & Jeffreys  & 0.9674 & 0.9494 & 0.9715 & 0.9636  & 0.9661     \\ 
           & Reference & 0.9338 & 0.9494 & 0.9715 & 0.9518  & 0.9447     \\\hline
           & MLE       & 0.9971 & 0.9940 & 0.9438 & 0.9218  & 0.9453     \\ 
           & CMLE      & 0.9971 & 0.9940 & 0.9438 & 0.9218  & 0.9453     \\ 
$\alpha_2$  & Jeffreys & 0.9203 & 0.9514 & 0.9365 & 0.9481  & 0.9497     \\ 
           & Reference & 0.9707 & 0.9514 & 0.9526 & 0.9439  & 0.9552     \\\hline  \hline 
\end{tabular}
\label{tableres1d}
\end{table*}

Finally, the reference prior using Theorem \ref{theoveralpri} is
\begin{equation}\label{genpriref}
\pi^R (\boldsymbol\theta) \propto \prod_{j=1}^{p} \beta_j^{-1} \alpha_j^{-1/2} .
\end{equation}
Thus, in this case, the joint posterior distribution is
\begin{eqnarray}\label{genpostref}
\pi^R(\boldsymbol\theta | \boldsymbol t) \propto \prod_{j=1}^{p} \gamma(\beta_j | n_j, n_j / \hat{\beta_j}) \gamma( \alpha_j | n_j + \tfrac{1}{2} , 1).
\end{eqnarray}

\begin{proposition} The overall reference prior (\ref{genpostref}) is matching prior for all parameters. 
\end{proposition}
\begin{proof} The proof is essentially the same as that 
of Proposition \ref{propov2cmat}. 
\end{proof}

Using the same approach of the case where $\beta_1\neq\beta_2$, 
the Bayes estimators are 
\begin{equation}
\hat{\beta_j}^{Bayes} = \left(\frac{n_j-1}{n_j}\right)\hat{\beta_j}^{MLE}
\end{equation}
and, for $j = 1 , \ldots , p$, 
\begin{equation}
\hat\alpha_j^{Bayes}=n_j.
\end{equation}

\section{Simulation Study}

In this section we present a simulation study 
to compare the Bayes estimators and the MLEs. 
We used two criteria to evaluate the estimators behaviour: 
the mean relative error (MRE) 
\begin{equation*}
\f{MRE}_{\hat{\theta}_i}=\frac{1}{M}\sum_{j=1}^{M}\frac{\hat\theta_{i,j}}{\theta_i} 
\end{equation*} 
and the mean square error (MSE)
\begin{equation*}
\f{MSE}_{\hat{\theta}_i} = 
\sum_{j=1}^{M}\frac{(\hat\theta_{i,j}-\theta_i)^2}{M},
\end{equation*} 
where $M$ is 
is the number of estimates 
(i.e.\ the Monte Carlo size), which we take 
$M=1,000,000$ throughout the section,
and $\boldsymbol{\theta}=(\theta_1,\ldots,\theta_p)$ is the vector of parameters. 

Additionally, for the objective bayesian credibility intervals 
and the asymptotic maximum likelihood based confidence intervals
for $\beta_1$, $\beta_2$, $\alpha_1$ and $\alpha_2$
we computed the $95\%$ interval coverage probability, 
denoted by $CP_{95\%}$. 
Good estimators should have 
MRE close to one and MSE close to zero and  
good intervals should be short while showing 
$CP_{95\%}$ close to 95\%.
We also considered a confidence interval based 
on the CMLE obtained from the unbiased estimator (\ref{cmleeq}) 
and the asymptotic variances estimated from the
Fisher information matrix, similar to what is done 
to obtain the ML interval.
 
The results were computed using the software R. 
Below we show the results 
for a single system subject to two causes of failure. 
We assumed that the two-component system was observed on the fixed time interval $[0, T]$ and
considered five different scenarios for $T$ and 
the parameters: 

\begin{itemize}
\item Scenario 1: $\beta_1=1.5, \alpha_1=6.45, \beta_2=1.0, \alpha_2=2.75, T=5.5$;
\item Scenario 2: $\beta_1=1.75, \alpha_1=26.46, \beta_2=1.25, \alpha_2=3.11, T=6.5$;
\item Scenario 3: $\beta_1=1.5,  \alpha_1=5.59, \beta_2=0.8, \alpha_2=14.50, T=5.0$;
\item Scenario 4: $\beta_1=1.6,  \alpha_1=6.59, \beta_2=0.7, \alpha_2=15.12, T=5.0 $;
\item Scenario 5: $\beta_1=0.25, \alpha_1=8.46, \beta_2=2.0, \alpha_2=100.00, T=20.0 $;
\end{itemize}

The values of the parameters were selected in order to obtain different samples sizes. The results were presented only for these five scenarios due to the lack of space. However, the obtained results are similar for other choices of the parameters. Using the fact that the causes are independent and 
well known results about NHPPs 
\cite{rigdon2000statistical}, 
for each Monte Carlo replication the
failure times were generated as follows:
\begin{itemize}
\item[] \textbf{Step 1}: 
For each cause of failure, 
generate random numbers $n_j \sim Poisson(\Lambda_j)$ ($j=1,2$).
\item[] \textbf{Step 2}: 
For each cause of failure, 
let the failure times be $t_{1,j} , \ldots , t_{n_j , j}$, 
where $t_{i,j} = T \, U_{i,j}^{1/\beta_j}$ and 
$U_{1,j} , \ldots , U_{n_j , j}$ are the order 
statistics of a size $n_j$ random sample from 
the standard Uniform distribution.
\item[] \textbf{Step 3}: Finally, to get the data 
in the form $(t_i , \delta_i)$, 
let the $t_i$s be the set of ordered 
failure times and set $\delta_i$ equal to 1 or 2 
according to the corresponding cause of failure 
(i.e.\ set $\delta_i =1$ if $t_i = t_{h,1}$ for some $h$ 
and $\delta_i =2$ otherwise).
\end{itemize}

Tables \ref{tableres1a} - \ref{tableres1d} present the results.
In Table \ref{tableres1a} the estimators for $\beta_j$ and $\alpha_j, j=1,2,$ are the same for the CMLE, Jeffreys and reference, we denote as Bayes estimators. 
However, they are different when computing the CIs and the CPs.

We note from Table \ref{tableres1a} 
that the Bayes estimator behaves consistently better 
than the MLE across the different scenarios, and that 
this holds both for the MRE and the MSE criteria.

Regarding the coverage probabilities in Table \ref{tableres1d}, 
we note that for both the MLE and the CMLE CIs the CP are far 
from the assumed levels, 
especially for the shape parameters $\alpha_1$ and $\alpha_2$. 
On the other hand, the CP of the Bayes estimators using the reference posterior returned accurate coverage probabilities. 
These results were expected as we had proved that the overall reference prior is a matching prior for all the parameters. 


\vspace{0.3cm}
\section{Real Data Applications}
\vspace{0.3cm}

Below we analyze the two data sets described Section II.


\subsection*{sugarcane harvester}
\vspace{0.3cm}

There were 10 failures attributed to Cause 1, 24 to Cause 2 and 14 to Cause 3. Figure \ref{figureA} suggests a high recurrence of failures. 
The harvester has a history of intense breakdown in an observation window of only 8 months. This behavior may be explained 
by the intense use of the machine. 

\begin{figure}[!h]
	\centering
	\includegraphics[width=8.5cm]{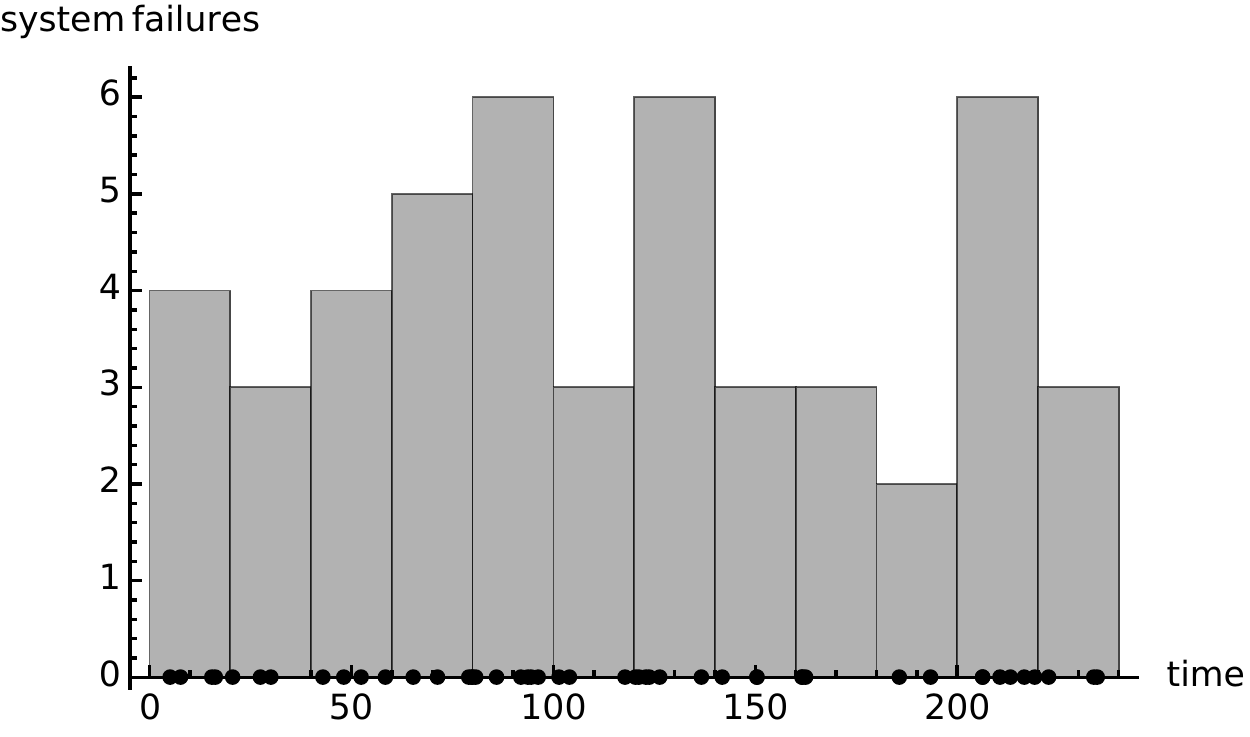}
	\caption{The black points on the x-axis indicates the system failures; the histogram shows the number of failures in each 20 days interval.}\label{figureA}
\end{figure}

We assessed initially the adequacy of the PLP for each cause of failure with the help of a Duane plot; see  \cite{rigdon2000statistical}. 
Figure \ref{figureB} shows plots of the 
logarithm of number of failures $N_j (t)$ against the
logarithm of accumulated mileages.  
Since the three plots exhibit a reasonable linearity, 
they suggest that the PLP model is adequate. 

\begin{figure}[!h]
	\centering
	\includegraphics[width=8cm]{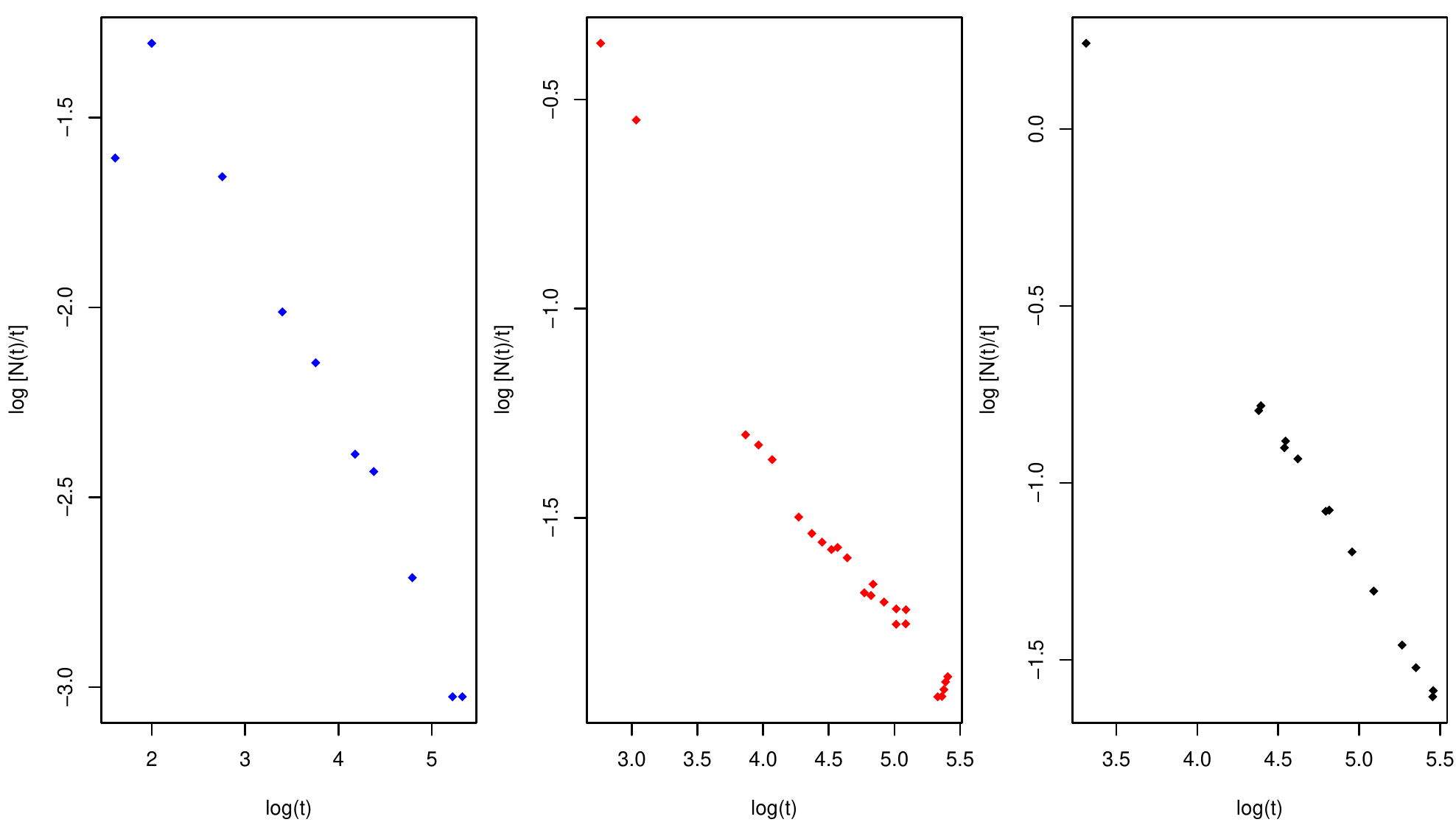}
	\caption{Duane plots: Cause 1 (blue), Cause 2 (red), Cause 3 (black).} \label{figureB}
\end{figure}

We summarize here the results concerning the objective Bayesian inference using the reference prior (\ref{genpriref}).
The Bayes estimates are shown in Table \ref{tableW2}, along with the corresponding marginal posterior standard deviations (SD) and credible intervals (CI). We remark that 
this posterior summary does not require stochastic simulation to be obtained. For instance, the credible intervals can be calculated directly from the posterior quantiles from (\ref{genpostref}).

\begin{table}[!h]
\renewcommand{\arraystretch}{1.3}
\caption{Bayesian estimators using reference posterior for harvester data}
\centering
\begin{tabular}{crcc}
\hline\hline
Parameter & Bayes & SD & CI (95\%)\\ 
\hline\hline
$\beta_1$ & 0.553 & 0.175 & [0.265 ; 0.945]  \\ 
$\beta_2$ & 1.079 & 0.220 & [0.691 ; 1.551] \\ 
$\beta_3$ & 1.307 & 0.349 & [0.714 ; 2.075] \\ 
$\alpha_1$ & 10.000 & 3.162 & [5.141 ; 17.739] \\ 
$\alpha_2$ & 24.000 & 4.899 & [15.777 ; 35.111] \\ 
$\alpha_3$ & 14.000 & 3.742 & [8.024 ; 22.861] \\ 
\hline \hline 
\end{tabular}\label{tableW2}
\end{table}

The results suggest that the the reliability 
of the electrical components (Cause 1) is 
improving with time, since the corresponding 
$\hat{\beta}_1 = 0.553 < 1$, while the reliability 
of the elevator is decreasing ($\hat{\beta}_3 = 1.307 > 1$). 
The engine shows an intermediate behavior, since $\hat{\beta}_2 = 1.079$ 
is slightly greater than one). We remark that this information 
can provide important insights to the maintenance crew.

\subsection*{Automotive warranty claims data}
\vspace{0.3cm}

The dataset contains 99 failures attributed to Cause 1, 
118 to Cause 2 and 155 to Cause 3. 
The system failure distribution shown in 
Figure \ref{figureC} suggests that 
the number of failures decreases as time increases 
and, in fact, after 2000 miles there is a sharp reduction 
of failures due to all three causes.

\begin{figure}[!h]
	\centering
	\includegraphics[width=8.5cm]{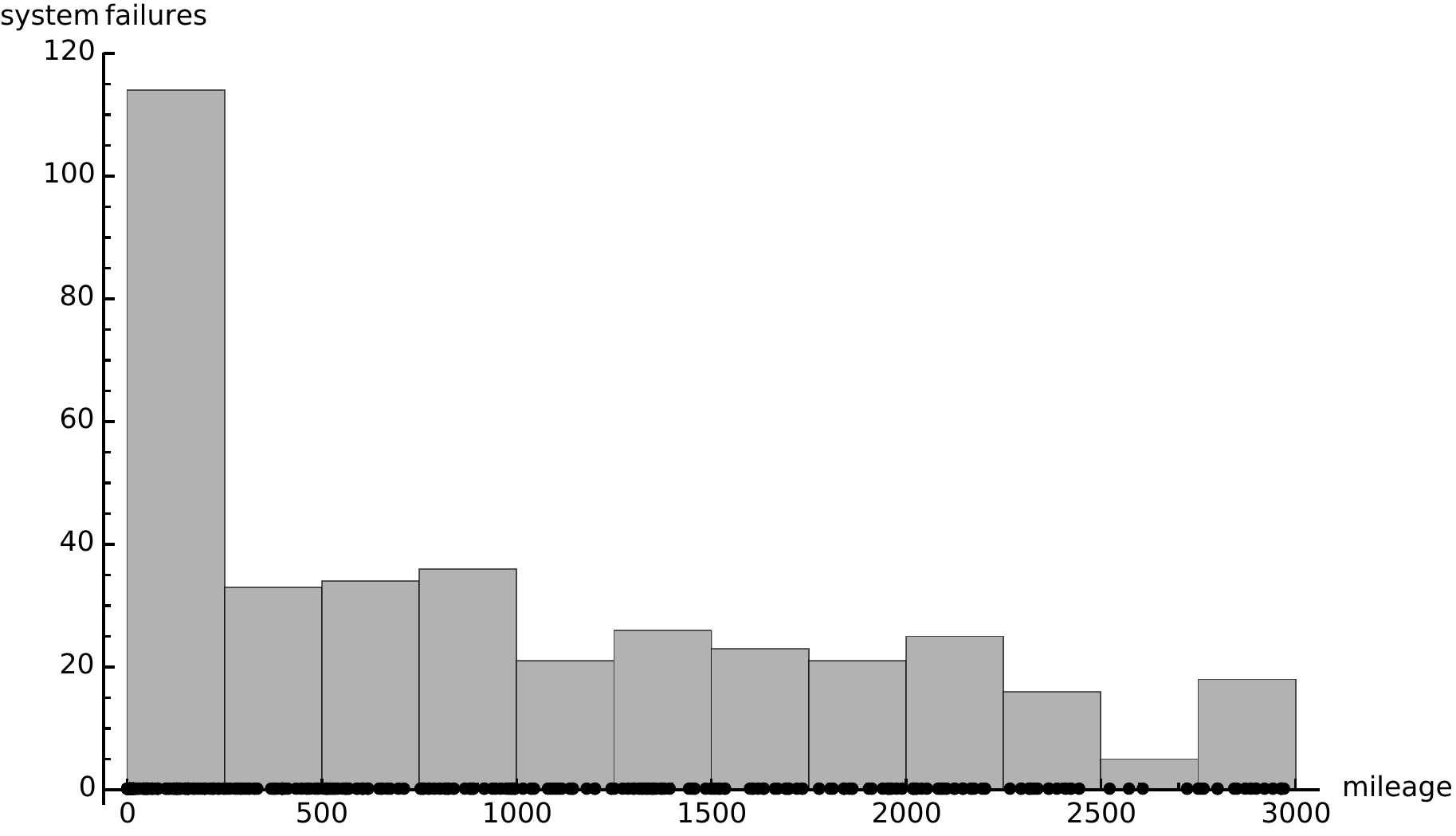}
	\caption{The black points on the x-axis indicates the system failures; the histogram shows the number of failures in each 250-mileage interval.}\label{figureC}
\end{figure}	

Figure \ref{figureD} shows the Duane plots for the three causes of failures. 
It suggest that, at least for Causes 2 and 3, the PLP should fit the 
data well. The plot corresponding to Cause 1 is less conclusive. 
Notwithstanding, we show below the results assuming the PLP 
for all three causes.

\begin{figure}[!h]
	\centering
	\includegraphics[width=8cm]{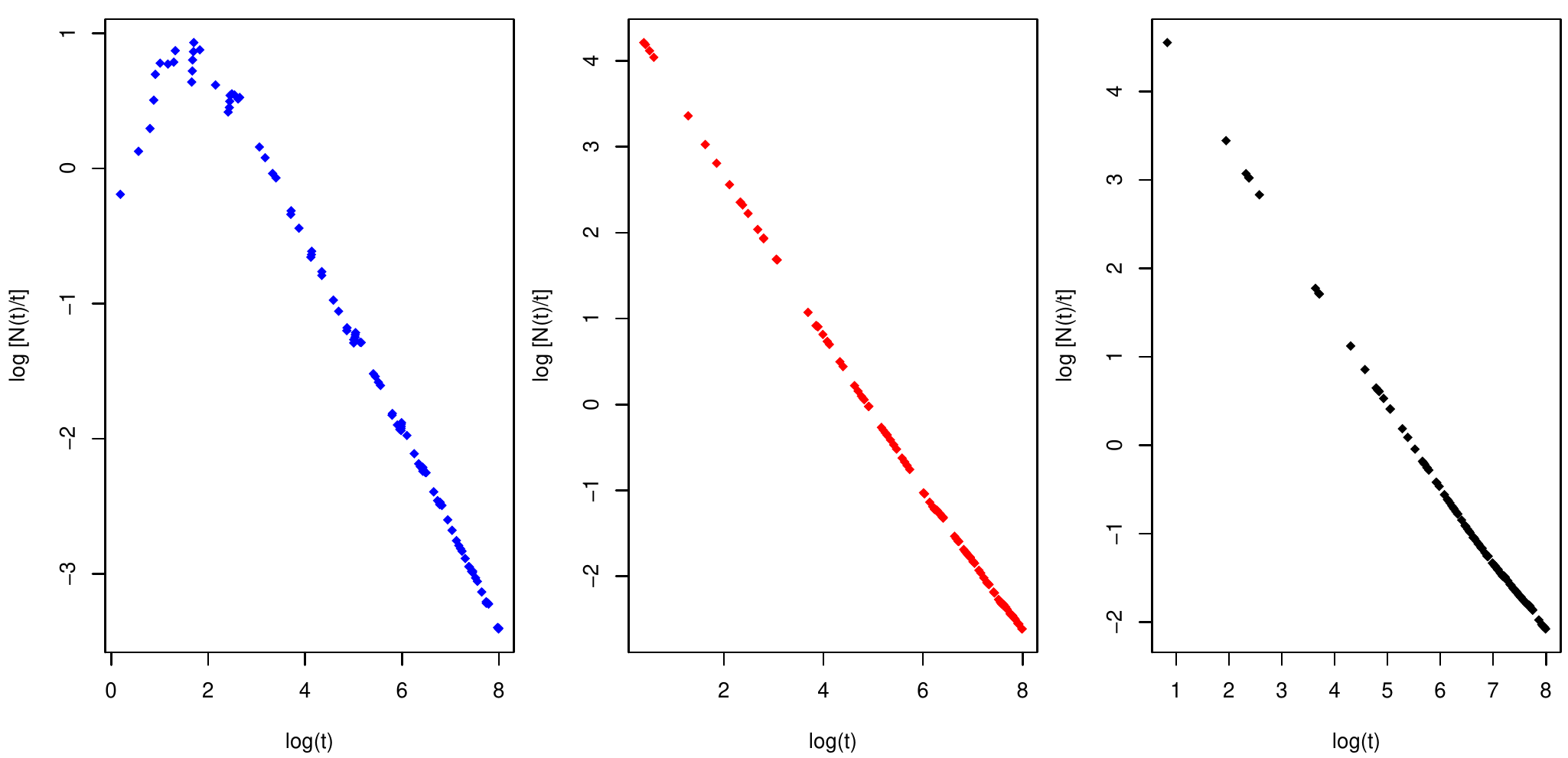}
	\caption{Duane plots: Cause 1 (blue), Cause 2 (red), Cause 3 (black).} \label{figureD}
\end{figure}

Table \ref{tableW} shows the estimates and credible intervals 
assuming the reference prior (\ref{genpriref}).

\begin{table}[!h]
\renewcommand{\arraystretch}{1.3}
\caption{Bayesian estimators using reference posterior for warranty claim data}
\centering
\begin{tabular}{crcc}
\hline\hline
Parameter & Bayes & SD & CI (95\%)\\ 
\hline\hline
$\beta_1$ & 0,329 & 0,033 & [0.267 ; 0.396] \\ 
$\beta_2$ & 0,451 & 0,042 & [0.374 ; 0.536] \\ 
$\beta_3$ & 0,721 & 0.058 & [0.612 ; 0.839] \\ 
$\alpha_1$ & 99.000 & 9.950 & [80.913 ; 119.980] \\ 
$\alpha_2$ & 118.000 & 10.863 & [98.127 ; 140.767] \\ 
$\alpha_3$ & 155.000 & 12.450 & [132.020 ; 180.874] \\ 
\hline \hline 
\end{tabular}
\label{tableW}
\end{table}

The obtained results are more precise than the presented in Somboonsavatdee and Sen \cite{somboonsavatdee2015statistical} as they are unbiased estimates. Moreover, the credibility intervals have more accurate nominal levels.

\section{Discussion}

The proposal to model a single reparable system under the assumption of minimal repair with competing risks has some gaps to be filled from the following points of view: Firstly, few studies have considered multiple 
failure causes from the perspective of repairable systems and cause-specific intensity functions. 
Secondly, Bayesian methods have been not well explored in this context. 
The competing risks approach may be advantageous 
in the engineering field because it may lead to a better understanding  
of the various causes of failure of a system and hence  
design strategies to improve the overall reliability.

We considered the competing risk approach in reparable systems under the action of multiple causes of failure, assuming that the causes act independently.
We proposed Bayesian estimates for the parameters of the 
failure intensity assuming the PLP model and using 
an reference posterior. We showed that the resulting marginal posterior intervals have accurate coverage in the frequentist sense. Moreover, the obtained posterior is proper and has interesting properties, such as one-to-one invariance, consistent marginalization, and consistent sampling properties. Moreover, the 
Bayes estimates have closed-form expressions and are naturally unbiased. An simulation study suggest that they 
outperform the estimates obtained from the ML approach.

The proposed methodology was applied to an original data set regarding failures of a sugarcane harvester 
classified according to three possible causes. Since the data contains few failures, classical CIs based on asymptotic ML theory could be inadequate in this case. 

There may be some interesting extensions of this work. 
One can consider a maintenance analysis scenario whose objective is to determine the optimal periodicity of preventive maintenance, as presented for instance in Oliveira et al. \cite{de2012bayesian}. 
More challenging could be to investigate dependency structures 
between the causes introducing, for instance, frailty terms in the model.

\section*{Acknowledgment}

The authors would like to thank.

\ifCLASSOPTIONcaptionsoff
  \newpage
\fi

\bibliographystyle{IEEEtran}

\bibliography{referencias}

\end{document}